\documentclass[lettersize,onecolumn]{IEEEtran}
\usepackage{amsmath,amsfonts}
\usepackage{algorithmic}
\usepackage{algorithm}
\usepackage{array}
\usepackage[caption=false,font=normalsize,labelfont=sf,textfont=sf]{subfig}
\usepackage{textcomp}
\usepackage{stfloats}
\usepackage{url}
\usepackage{verbatim}
\usepackage{graphicx}
\usepackage{cite}
\usepackage{amssymb}
\usepackage{mathrsfs}
\usepackage{color}
\usepackage{amsthm}
\usepackage{setspace}
\usepackage{rotating}
\usepackage{mathdots}
\usepackage{bm}
\begin{document}

\title{ Hermitian Self-dual Generalized Reed-Solomon Codes}%\\ for IEEE Journals and Transactions}

\author{Chun'e Zhao$^{1}$, Wenping Ma$^{2}$\\
\begin{spacing}{2.0}
\end{spacing}
\small{$1$ College of Sciences,
China University of Petroleum,
Qingdao 266555,
Shandong, China\\
$2$ School of Telecommunication Engineering,
 Xidian University,
  Xi'an, China\\
Email: zhaochune1981@163.com;\ wp\_ma@mail.xidian.edu.cn;\\}
}%~\IEEEmembership{Staff,~IEEE,}}
% <-this % stops a space
%\thanks{Chun'e Zhao is with the Faculty of College of Sciences, China University of Petroleum, Qingdao, Shandong, China }

%Shandong, China}}% <-this % stops a space
%\thanks{Wenping Ma is with the faculty of School of Telecommunication Engineering, Xidian University, Xi'an, China}

%\thanks{Tongjiang Yan is with the Faculty of College of Sciences, China University of Petroleum, Qingdao, Shandong, China }

%\thanks{Yuhua Sun is with the Faculty of College of Sciences, China University of Petroleum, Qingdao, Shandong, China }}

%\markboth{Journal of \LaTeX\ Class Files,~Vol.~1, No.~2, December~2023}%

%\IEEEpubid{0000--0000~\copyright~2023 IEEE}

\maketitle
\begin{abstract}
 Maximum Distance Separable (MDS) self-dual codes are of significant theoretical and practical importance. Generalized Reed-Solomon (GRS) codes are the most prominent MDS codes. Correspondingly there have been many research on constructions of Euclidean self-dual MDS codes by using GRS codes. However, the study on Hermitian self-dual GRS codes is relatively limited. Since Hermitian self-dual GRS codes do not exist for $n>q+1$,  this paper is devoted to an investigation of GRS codes in the case where $n\le q+1$. First, we prove that when $n\leq q+1$, there are only two classes of Hermitian self-dual GRS codes, confirming the conjecture in \cite{2019yue} and providing its proof simultaneously. Second, we present two explicit construction methods. Thus, the  existence and construction of Hermitian self-dual GRS codes are fully solved.

\end{abstract}


\section{Introduction}

\IEEEPARstart{L}et $\mathbb{F}_{q}$ be a finite filed of size $q$, where $q$ is a prime power and $\mathbb{F}_{q}^{*}=\mathbb{F}_{q}\backslash \{0\}$. An $[n, k, d]$ linear code $C$ is a subspace of the linear space $\mathbb{F}_ {q}^{n}$ over $\mathbb{F}_ {q}$ with dimension $k$ and minimum Hamming distance $d$. It is well known that the parameters of $C$ satisfy $d\leq n-k+1$. If $d=n-k+1$, $C$ is called a maximum distance separable (MDS) code.

For a linear code $C$, we use $C^{\perp E}$ and $C^{\perp H}$ to denote the Euclidean dual code of $C$ and the Hermitian dual code of $C$, respectively. If $C=C^{\perp E}$ , then the linear code $C$ is called Euclidean self-dual. And if $C=C^{\perp H}$, then $C$ is called Hermitian self-dual. 

Both MDS codes and self-dual codes possess extensive and significant applications across numerous research and engineering domains. Due to the maximum error-correcting capability, MDS codes are  widely used in various data storage systems, such as the coding design of distributed storage systems \cite{2018}. Meanwhile,  MDS codes also have close connections with mathematical branches such as combinatorial design and finite geometry \cite{1977}. Self-dual codes, on the other hand, show important value in fields such as cryptography \cite{2008-1}\cite{2008-2}and lattice theory \cite{1999}\cite{2003}. A linear code is called an MDS self-dual code if it is both MDS and self-dual.

In recent years, the investigation of MDS sel-dual codes has become a hot topic in algebrac coding theory. Extensive research  has been devoted to the study of Euclidean or Hermitian self-dual MDS codes. Kim and Lee proposed an efficient construction method for self-dual codes with respect to the Euclidean or Hermitian inner product and constructed many corresponding self-dual MDS codes in \cite{2004-1}. Guenda established a systematic construction for Euclidean and Hermitian self-dual MDS codes via  extended cyclic duadic codes or negacyclic codes in \cite{2012}. 

The GRS codes are the most prominent MDS codes. It is meaningful to construct MDS self-dual codes via GRS codes. A disciplined construction of Euclidean self-dual MDS codes utilizing GRS codes was first proposed by Jin and Xing in \cite{2017jinlingfei}. Subsequently,  GRS codes become one of the most prevalent and effective tools to construct MDS self-dual codes. Yan established a necessary and sufficient condition under which the GRS code is an Euclidean self-dual MDS code in \cite{2018-2}. Labad et al. constructed more classes of Euclidean self-dual MDS codes in \cite{2019liu}. Fang W. et al. \cite{2019fang} constructed six additional classes of q-ary MDS Euclidean self-dual codes by using GRS codes and extended GRS codes. Fang X. et al. \cite{2020liu} produced several classes of Euclidean self-dual MDS codes via (extended) GRS codes. Ning Y. et al. \cite{2021ge} focuses on constructions of MDS  Euclidean self-dual codes from (extended) GRS codes and extended the consecutive range. Fang W. et al. in \cite{2022fang} gave a systematic way to construct   Euclidean self-dual GRS codes with flexible evaluation points. Zhang and Feng presented a unified approach on the existence of MDS  Euclidean self-dual codes in \cite{2020feng}. Thus the research on Euclidean self-dual GRS codes has become increasingly abundant. However, the study of Hermitian self-dual GRS codes is relatively scarce. In 2019, Niu and Yue et al. in \cite{2019yue} provided two classes of the Hermitian self-dual MDS GRS codes and proved that Hermitian self-dual MDS GRS codes must be above two classes for $n=4$ and conjectured that the result also holds for each even length $n$, $4\leq n\leq q+1$. In 2021, Guo and Li  \cite{2021li} revised and improved the results in \cite{2019yue} and conjectured that there are no Hermitian self-dual GRS codes when $n>q+1$ by Magma programming. In 2025, Wan and Zhu theoretically proved it in \cite{2025zhu}. Besides this, there is no other relevant introduction to Hermitian self-dual GRS codes.
 
In order to elaborate on the content of this subject, this paper focuses on Hermitian self-dual MDS codes and GRS codes for $n\leq q+1$. Specifically, we demonstrate that only two classes of Hermitian self-dual GRS codes exist for $n\leq q+1$; this result not only verifies the correctness of the conjecture proposed in \cite{2019yue} but also completes its rigorous proof. Furthermore, two specific constructions for such Hermitian self-dual GRS codes are provided herein. Consequently, the research on the existence and constructions of Hermitian self-dual GRS codes is completely resolved.

The organization of the rest of the paper is as follows. In Section 2, we will introduce some basic knowledge and auxiliary results on GRS codes and Hermitian self-dual codes. In Section 3, we will present our main results on Hermitian self-dual GRS codes.
In Section 4, we will give two specific constructions of MDS Hermitian self-dual codes by using GRS codes.

\section{Preliminaries}
In this section, we introduce basic definitions of Hermitian self-dual codes and GRS codes, which will be employed in our sebsequent discussion.
\subsection{ Hermitian Self-Dual Codes}
Let $q$ be a prime power and $\mathbb{F}_{q^{2}}$ be a finite field with $q^{2}$ elements . Suppose $\bm{x}=(x_{1},x_{2},\cdots,x_{n})$ and $\bm{y}=(y_{1},y_{2},\cdots,y_{n})$ are two vectors in $\mathbb{F}_{q^{2}}^{n}$. There are two inner products as follows.

The Euclidean inner product of $\bm{x}$ and $\bm{y}$ is defined as:
$$<\bm{x},\bm{y}>_{E}=x_{1}y_{1}+x_{2}y_{2}+\cdots+x_{n}y_{n}.$$
The Hermitian inner product of $\bm{x}$ and $\bm{y}$ is defined as:
$$<\bm{x},\bm{y}>_{H}=x_{1}y_{1}^{q}+x_{2}y_{2}^{q}+\cdots+x_{n}y_{n}^{q}.$$
Let $C$ be a linear code of length $n$ over $\mathbb{F}_{q^{2}}$. The Euclidean dual of $C$, denoted $C^{\perp E}$, is defined as
\begin{center}
   $C^{\perp E}=\{\bm{x}\in \mathbb{F}_{q^2}^{n}|<\bm{x},\bm{y}>_{E}=0, \text{ for all } \bm{y}\in C\}$.
\end{center}
Similarly, the Hermitian dual code of $C$, denoted $C^{\perp H}$, is defined as
\begin{center}
   $C^{\perp H}=\{\bm{x}\in \mathbb{F}_{q^2}^{n}|<\bm{x},\bm{y}>_{H}=0, \text{ for all } \bm{y}\in C\}$.  
\end{center}
 If $C=C^{\perp H}$, then $C$ is called Hermitian self-dual code.
%For vector $\bm{v}=(v_{1},v_{2},\cdots,v_{n})\in \mathbb{F}_{q^2}^{n}$, $\bm{v}^{q}$ is defined to be the vector $(v_{1}^{q},v_{2}^{q},\cdots,v_{n}^{q})$. For a subset $V$ of 
%$\mathbb{F}_{q^2}^{n}$, we define $V^{q}$ to be the set $\{\bm{v}^{q}|\bm{v}\in V\}$. It is easy to verify that $C^{\perp H}=(C^{q})^{\perp E}$. Note that the code $C$ is a Hermitian self-dual code if and only if $C=(C^{q})^{\perp E}$. Namely $C^{q}=C^{\perp E}$.
\subsection{ GRS Codes}
Let $k,n$ be positive integers with $k\leq n$. Define the set 
$$\mathbb{F}_{q^{2}}[x]_{k}=\{f(x)\in \mathbb{F}_{q^{2}}[x]|deg(f(x))\leq k-1\}.$$
It is easy to verify that $\mathbb{F}_{q^{2}}[x]_{k}$ is a $k-$dimensional linear space over $\mathbb{F}_{q^2}$. Let $\bm{\alpha}=(\alpha_{1},\alpha_{2},\cdots,\alpha_{n})\in (\mathbb{F}_{q^2})^{n}$, where $\alpha_{1},\alpha_{2},\cdots,\alpha_{n}$ are $n$ distinct elements in $\mathbb{F}_{q^2}$, and $\bm{v}=(v_{1},v_{2},\cdots,v_{n})\in (\mathbb{F}_{q^2}^{*})^{n}$. Then the generalized Reed-Solomon code, $GRS_{n,k}(\bm{\alpha},\bm{v})$, associated with $\bm{\alpha}$ and $\bm{v}$ is defined by 
$$GRS_{n,k}(\bm{\alpha},\bm{v})=\{(v_{1}f(\alpha_{1}),v_{2}f(\alpha_{2}),\cdots,v_{n}f(\alpha_{n}))|\text{ for all }  f(x)\in \mathbb{F}_{q^{2}}[x]_{k}\}.$$
In fact $GRS_{n,k}({\bm\alpha},{\bm v})$ can also be viewed as a linear code with the following generator matrix:
$$G=[I_{k},0]V_{n}(\bm{\alpha})diag(\bm{v}),$$ where 
\begin{equation}\label{v matirx}
 I_{k}=\left(\begin{array}{cccc}
            1 &   &        &   \\
              & 1 &        &   \\
              &   & \ddots &   \\
              &   &        & 1
        \end{array}
    \right),V_{n}({\bm \alpha})=\left(\begin{array}{cccc}
            1                & 1                & \cdots                   & 1                \\
            \alpha_{1}       & \alpha_{2}       & \cdots                   & \alpha_{n}       \\
            \vdots           & \vdots           & \textcolor{blue}{\ddots} & \vdots           \\
            \alpha_{1}^{n-1} & \alpha_{2}^{n-1} & \cdots                   & \alpha_{n}^{n-1}
        \end{array}\right),\mathrm{diag}({\bm v})=\left(\begin{array}{cccc}
            v_{1} &       &        &       \\
                  & v_{2} &        &       \\
                  &       & \ddots &       \\
                  &       &        & v_{n}
        \end{array}\right).
\end{equation}
%As we know, the code $GRS_{n,k}(\bm{\alpha},\bm{v})$ is an MDS code over $\mathbb{F}_{q^{2}}$ with parameters $[n,k,n-k+1]$. Furthermore, if the code $GRS_{n,k}(\bm{\alpha},\bm{v})$ is MDS self-dual code then $n$ is even and $k=\frac{n}{2}$. The generator matrix of the code $GRS_{k}(\bm{\alpha},\bm{v})$ is defined as:
%$$G_{k}(\bm{\alpha},\bm{v})=\left(\begin{array}{cccc}
%     v_{1}\alpha_{1}^{0}& v_{2}\alpha_{2}^{0}&\cdots&v_{n}\alpha_{n}^{0} \\
%     v_{1}\alpha_{1}^{1}& v_{2}\alpha_{2}^{1}&\cdots&v_{n}\alpha_{n}^{1} \\  
%     \vdots&\vdots&&\vdots\\
%       v_{1}\alpha_{1}^{k-1}& v_{2}\alpha_{2}^{k-1}&\cdots&v_{n}\alpha_{n}^{k-1} 
 %   \end{array}
 %   \right)$$
    
\newtheorem{definition}{Definition}

\section{Hermitian Self-dual GRS codes}
In this section, we initiate our study of Hermitian self-dual GRS codes by employing matrix theory and linear feedback shift register sequence theory. We first introduce some relevant notation that will be used throughout our discussion.

\begin{itemize}
\item[(1)]  For any vector $\bm{\xi}\in\mathbb{F}_{q^{2}}^{n}$, $\bm{\xi}^{T}$ represents the transpose of $\bm{\xi}$.\\
    \item [(2)] For any sequence $S=(s_{0},s_{1},s_{2},\cdots)$, $L^{(i)}(S)=(s_{i},s_{i+1},s_{i+2},\cdots)$, $L^{(i)}(S)_{[j]}=(s_{i},s_{i+1},s_{i+2},\cdots,s_{i+j-1})$.\\
    \item[(3)] For any matrix $A_{m\times n}$, $A[i_{1},i_{2},\cdots,i_{t};j_{1},j_{2},\cdots,j_{s}]$ represents the sub matrix by rows $i_{1},i_{2},\cdots,i_{t}$ and columns $j_{1},j_{2},\cdots,j_{s}$ of $A$, where $0\leq i_{k}\leq m-1$ and $0\leq j_{k}\leq n-1$.
\end{itemize}

\newtheorem{theorem}{Theorem}
\begin{theorem}\label{State transition matrix}
    Let $\alpha_{1},\alpha_{2},\cdots,\alpha_{n}$ be distinct elements in the finite field $\mathbb{F}_{q^{2}}$, where $q$ is a prime power. Let $G(x)=\prod\limits_{i=1}^{n}(x-\alpha_{i})=x^{n}-\sum\limits_{i=0}^{n-1}c_{i}x^{i}$, then the matrix 
     \begin{equation}\label{statement transaction matrix}
     T=\left(\begin{array}{ccccc}
    0&1&&&\\
    &0&1&&\\
    &&0&\ddots&\\
    &&&\ddots&1\\
    c_{0}&c_{1}&c_{2}&\cdots&c_{n-1}
    \end{array}
    \right)   
    \end{equation}
    has $\alpha_{1},\alpha_{2},\cdots,\alpha_{n}$ as its eigenvalues, and the eigenvector corresponding to each $\alpha_{i}$ is $(1,\alpha_{i},\alpha_{i}^{2},\cdots,\alpha_{i}^{n-1})^{T}$
  for $i=1,2,\cdots,n$.
\end{theorem}
$\ $
$\ $
\begin{proof}
The characteristic polynomial of the matrix $T$ is 
$$\left|xI_{n}-T\right|=\left|\begin{array}{ccccc}
    x&-1&&&\\
    &x&-1&&\\
    &&\ddots&\ddots&\\
    &&&x&-1\\
   -c_{0}&-c_{1}&-c_{2}&\cdots&x-c_{n-1}
    \end{array}\right|=x^{n}-c_{n-1}x^{n-1}-\cdots-c_{1}x-c_{0}=G(x),$$
where $I_{n}$ is defined as in Eq. (\ref{v matirx}). Thus, the eigenvalues of $T$ are $\alpha_{1},\alpha_{2},\cdots,\alpha_{n}$.
    Moreover, $\alpha_{i}^{n}=c_{n-1}\alpha_{i}^{n-1}+\cdots+c_{1}\alpha_{i}+c_{0}$, so
   
    $$T\left(\begin{array}{c}
    1\\
    \alpha_{i}\\
    \alpha_{i}^{2}\\
    \vdots\\
    \alpha_{i}^{n-1}\end{array}
    \right)=\left(\begin{array}{ccccc}
    0&1&&&\\
    &0&1&&\\
    &&0&\ddots&\\
    &&&\ddots&1\\
    c_{0}&c_{1}&c_{2}&\cdots&c_{n-1}
    \end{array}
    \right)\left(\begin{array}{c}
    1\\
    \alpha_{i}\\
    \alpha_{i}^{2}\\
    \vdots\\
    \alpha_{i}^{n-1}\end{array}
    \right)=\left(\begin{array}{c}
       \alpha_{i}\\
    \alpha_{i}^{2}\\
    \vdots\\
    \alpha_{i}^{n-1}\\
    \alpha_{i}^{n} \end{array}
    \right)=\alpha_{i}\left(\begin{array}{c}
    1\\
    \alpha_{i}\\
    \alpha_{i}^{2}\\
    \vdots\\
    \alpha_{i}^{n-1}\end{array}
   \right)$$ for $i=1,2,\cdots,n$.
\end{proof}
\begin{theorem}\label{Annihilating polynomial}
  Let $\alpha_{1},\alpha_{2},\cdots,\alpha_{n}$ be distinct elements in the finite field $\mathbb{F}_{q^{2}}$. Let $x_{1},x_{2},\cdots,x_{n}\in \mathbb{F}_{q^{2}}$, and define $\Delta_{i}=\sum\limits_{l=1}^{n}\alpha_{l}^{i}x_{l}$ for $i=0,1,\cdots$. If $\alpha_{1},\alpha_{2},\cdots,\alpha_{n}$ are the roots of $x^{m}+a_{m-1}x^{m-1}+\cdots+a_{1}x+a_{0}=0$, then 
  $$\Delta_{m+i}+a_{m-1}\Delta_{m-1+i}+\cdots+a_{1}\Delta_{1+i}+a_{0}\Delta_{i}=0$$for any non-negative integer $i$.
\end{theorem}
\begin{proof}
Since  $\alpha_{1},\alpha_{2},\cdots,\alpha_{n}$ are the roots of $x^{m}+a_{m-1}x^{m-1}+\cdots+a_{1}x+a_{0}=0$, then  $$\alpha_{j}^{m}+a_{m-1}\alpha_{j}^{m-1}+\cdots+a_{1}\alpha_{j}+a_{0}=0,j=1,2,\cdots,n.$$
i.e. $\alpha_{j}^{i}\left(\alpha_{j}^{m}+a_{m-1}\alpha_{j}^{m-1}+\cdots+a_{1}\alpha_{j}+a_{0}\right)=\alpha_{j}^{m+i}+a_{m-1}\alpha_{j}^{m-1+i}+\cdots+a_{1}\alpha_{j}^{1+i}+a_{0}\alpha_{j}^{i}=0$ for $j=1,2,\cdots,n$, i.e.
$$\left(a_{0},a_{1},\cdots,a_{m-1},1\right)\left(\begin{array}{cccc}
\alpha_{1}^{i}&\alpha_{2}^{i}&\cdots&\alpha_{n}^{i}\\
\alpha_{1}^{i+1}&\alpha_{2}^{i+1}&\cdots&\alpha_{n}^{i+1}\\
\vdots&\vdots&&\vdots\\
\alpha_{1}^{m-1+i}&\alpha_{2}^{m-1+i}&\cdots&\alpha_{n}^{m-1+i}\\
\alpha_{1}^{m+i}&\alpha_{2}^{m+i}&\cdots&\alpha_{n}^{m+i}
\end{array}
\right)=0.$$ Then $\Delta_{m+i}+a_{m-1}\Delta_{m-1+i}+\cdots+a_{1}\Delta_{1+i}+a_{0}\Delta_{i}=
\left(a_{0},a_{1},\cdots,a_{m-1},1\right)\left(\begin{array}{c}
\Delta_{i}\\
\Delta_{1+i}\\
\vdots\\
\Delta_{m-1+i}\\
\Delta_{m+i}
\end{array}\right)$
$$=\left(a_{0},a_{1},\cdots,a_{m-1},1\right)\left(\begin{array}{cccc}
\alpha_{1}^{i}&\alpha_{2}^{i}&\cdots&\alpha_{n}^{i}\\
\alpha_{1}^{1+i}&\alpha_{2}^{1+i}&\cdots&\alpha_{n}^{1+i}\\
\vdots&\vdots&&\vdots\\
\alpha_{1}^{m-1+i}&\alpha_{2}^{m-1+i}&\cdots&\alpha_{n}^{m-1+i}\\
\alpha_{1}^{m+i}&\alpha_{2}^{m+i}&\cdots&\alpha_{n}^{m+i}
\end{array}
\right)\left(\begin{array}{c}
x_{1}\\
x_{2}\\
\vdots\\
x_{n-1}\\
x_{n}
\end{array}\right)=0.$$
\end{proof}
By Theorems \ref{State transition matrix} and \ref{Annihilating polynomial}, we immediately obtain the following result.
\newtheorem{corollary}{Corollary}
\begin{corollary}\label{Statement trasaction}
 Let $\alpha_{1},\alpha_{2},\cdots,\alpha_{n}$ be distinct elements in the finite field $\mathbb{F}_{q^{2}}$. Let  $G(x)=\prod\limits_{i=1}^{n}(x-\alpha_{i})=x^{n}-\sum\limits_{i=0}^{n-1}c_{i}x^{i}$, and let $x_{1},x_{2},\cdots,x_{n}\in \mathbb{F}_{q^{2}}$. Define $\Delta_{i}=\sum\limits_{l=1}^{n}\alpha_{l}^{i}x_{l},i=0,1,\cdots$. Let the matrix $T$ be defined as in Eq. (\ref{statement transaction matrix}).
    Then $$T\left(\begin{array}{c}
   \Delta_{i}\\
    \Delta_{i+1}\\
     \vdots\\
      \Delta_{i+n-1}\end{array}\right)=\left(\begin{array}{c}
   \Delta_{i+1}\\
    \Delta_{i+2}\\
     \vdots\\
     \Delta_{i+n}\end{array}\right).$$
\end{corollary}
\begin{proof}
 It is obvious that $\alpha_{1},\alpha_{2},\cdots,\alpha_{n}$ are the roots of $G(x)=0$. So 
 $$\alpha_{j}^{n}-c_{n-1}\alpha_{j}^{n-1}-\cdots-c_{1}\alpha_{j}-c_{0}=0,j=1,2,\cdots,n.$$
 By Theorem \ref{Annihilating polynomial}, we have the following expression
 $$\Delta_{n+i}-c_{n-1}\Delta_{n-1+i}-\cdots-c_{1}\Delta_{1+i}-c_{0}\Delta_{i}=0.$$
\end{proof}
\newtheorem{lemma}{Lemma}
\begin{lemma}\label{nonzero solution}
Let $n=2k$ be even and let $\bm\alpha=(\alpha_{1},\alpha_{2},\cdots,\alpha_{n})\in  (\mathbb{F}_{q^{2}})^{n}$, where $\alpha_{1},\alpha_{2},\cdots,\alpha_{n}$ are distinct elements. Then there exists a vector $\bm v=(v_{1},v_{2},\cdots,v_{n})\in (\mathbb{F}_{q^{2}}^{*})^{n}$ such that $GRS_{n,k}(\bm\alpha,\bm v)$ is a $q^{2}$-ary Hermitian self-dual if and only if the system of equations   $\sum\limits_{l=1}^{n}\alpha_{l}^{i+jq}x_{l}=0,0\leq i,j\leq k-1$
has a solution $\bm x=(x_{1},x_{2},\cdots,x_{n})\in (\mathbb{F}_{q}^{*})^{n}.$   
\end{lemma}
\begin{proof}
 The code $GRS_{n,k}(\bm \alpha,\bm v)$ is a $q^2$-ary Hermitian self-dual GRS code if and only if 
 \begin{align*}
  \sum\limits_{l=1}^{n}(\alpha_{l}^{i}v_{l})(\alpha_{l}^{j}v_{l})^q=\sum\limits_{l=1}^{n}\alpha_{l}^{i+jq}v_{l}^{1+q}=0,0\leq i,j\leq k-1.
 \end{align*}
 That is equivalent to $$\sum\limits_{l=1}^{n}\alpha_{l}^{i+jq}x_{l}=0,0\leq i,j\leq k-1$$
 has a solution $\bm x=(x_{1},x_{2},\cdots,x_{n})\in (\mathbb{F}_{q}^{*})^{n}$, where $x_{i}=v_{i}^{1+q},1\leq i\leq n$.
\end{proof}

{\bf Remark 1:} When $\bm\alpha=(\alpha_{1},\alpha_{2},\cdots,\alpha_{n})\in  (\mathbb{F}_{q^{2}}^{*})^{n}$, Lemma 1 coincides with the Proposition 1 in \cite{2019yue}. In fact, this result also holds when  $\bm\alpha=(\alpha_{1},\alpha_{2},\cdots,\alpha_{n})\in  (\mathbb{F}_{q^{2}})^{n}$. For the  reader's convenience, we restate it here.

\begin{theorem}\label{linear independent 1}
    Let $\mathbb{F}_{q^{2}}$ be a finite field with $q^{2}$ elements, where  $q$ is a prime power.  Let $n=2k$ satisfy $n\leq q+k$. Let $\alpha_{1},\alpha_{2},\cdots,\alpha_{n}$ be distinct elements in $\mathbb{F}_{q^{2}}$, and let   $x_{1},x_{2},\cdots,x_{n}\in \mathbb{F}_{q^{2}}^{*}$. Define $\Delta_{i}=\sum\limits_{l=1}^{n}\alpha_{l}^{i}x_{l},i=0,1,\cdots$. If $\Delta_{i+jq}=0$ for $0\leq i,j\leq k-1$, then the vectors $$\beta_{0},\beta_{1},\cdots,\beta_{k-1}$$ are linearly independent, where $\beta_{l}=(\Delta_{lq+k},\Delta_{lq+k+1},\cdots,\Delta_{lq+n-1} )$ for $l=0,1,\cdots,k-1$.
\end{theorem}
\begin{proof}
Let $\Delta_{i}=\sum\limits_{l=1}^{n}\alpha_{l}^{i}x_{l},i=0,1,\cdots$. Then $\Delta_{i+jq}=0$ for $0\leq i,j\leq k-1$. Let $G(x)=\prod\limits_{i=1}^{n}(x-\alpha_{i})=x^{n}-\sum\limits_{i=0}^{n-1}c_{i}x^{i}$ and $T$ is defined as in Eq. (\ref{statement transaction matrix}). By Corollary \ref{Statement trasaction}, the sequence $S=\left(\Delta_{0},\Delta_{1},\Delta_{2},\cdots\right)$ can be viewed as a linear feedback shift register sequence with $(\Delta_{0},\Delta_{1},\cdots,\Delta_{n-1})$ as the initial state and $T$ as the state transition matrix.  Let $\widetilde{\beta}_{l}=(\Delta_{lq},\Delta_{lq+1}\cdots,\Delta_{lq+k-1},\Delta_{lq+k},\Delta_{lq+k+1},\cdots,\Delta_{lq+n-1} )$ for $l=0,1,\cdots,k-1$.
If there exists $b_{0},b_{1},\cdots,b_{k-1}\in \mathbb{F}_{q^{2}}$ such that 
\begin{equation}\label{linear combination 1}
 \sum\limits_{i=0}^{k-1}b_{i}\widetilde{\beta}_{i}=0 .  
\end{equation}
By Corollary \ref{Statement trasaction}, we know that 
$\widetilde{\beta}_{i}^{T}=T^{iq}\widetilde{\beta}_{0}^{T}$ for $i=0,1,\cdots,k-1$.
Then $$\sum\limits_{i=0}^{k-1}b_{i}\widetilde{\beta}_{i}^{T}=\sum\limits_{i=0}^{k-1}b_{i}(T^{q})^{i}\widetilde{\beta}_{0}^{T}=0.$$
Let $f(x)=b_{k-1}x^{k-1}+\cdots+b_{1}x+b_{0}$, by Theorem \ref{State transition matrix}, then $f(\alpha_{i}^{q})$ $i=1,2,\cdots,n$ are all the eigenvalues of $\sum\limits_{i=0}^{k-1}b_{i}(T^{q})^{i}=f(T^{q})$. And $(1,\alpha_{i},\alpha_{i}^{2},\cdots,\alpha_{i}^{n-1})^{T}$
is the eigenvector of matrix $f(T^{q})$ belonging to eigenvalue $f(\alpha_{i}^{q})$ for $i=1,2,\cdots,n$.
    
    Then \begin{align*}
     \sum\limits_{i=0}^{k-1}b_{i}(T^{q})^{i}\widetilde{\beta}_{0}^{T}&=f(T^{q})\left(\begin{array}{c}
   \Delta_{0}\\
    \Delta_{1}\\
     \vdots\\
      \Delta_{n-1}\end{array}\right)\\
      &=f(T^{q})\left(\begin{array}{cccc}
                1                & 1                & \cdots & 1                \\
                \alpha_{1}       & \alpha_{2}       & \cdots & \alpha_{n}       \\
                \vdots           & \vdots           & \ddots & \vdots           \\
                \alpha_{1}^{n-1} & \alpha_{2}^{n-1} & \cdots & \alpha_{n}^{n-1}
            \end{array}
        \right)\left(\begin{array}{c}
   x_{1}\\
    x_{2}\\
     \vdots\\
      x_{n}\end{array}\right)\\
      &=\left(\begin{array}{cccc}
                1                & 1                & \cdots & 1                \\
                \alpha_{1}       & \alpha_{2}       & \cdots & \alpha_{n}       \\
                \vdots           & \vdots           & \ddots & \vdots           \\
                \alpha_{1}^{n-1} & \alpha_{2}^{n-1} & \cdots & \alpha_{n}^{n-1}
            \end{array}
        \right)\left(\begin{array}{cccc}
        f( \alpha_{1}^{q})&&&\\
        &f( \alpha_{2}^{q})&&\\
        &&\ddots&\\
        &&&f( \alpha_{n}^{q})\end{array}
        \right)\left(\begin{array}{c}
   x_{1}\\
    x_{2}\\
     \vdots\\
      x_{n}\end{array}\right)\\
      &=\left(\begin{array}{cccc}
                1                & 1                & \cdots & 1                \\
                \alpha_{1}       & \alpha_{2}       & \cdots & \alpha_{n}       \\
                \vdots           & \vdots           & \ddots & \vdots           \\
                \alpha_{1}^{n-1} & \alpha_{2}^{n-1} & \cdots & \alpha_{n}^{n-1}
            \end{array}
        \right)\left(\begin{array}{c}
    f( \alpha_{1}^{q})x_{1}\\
    f( \alpha_{2}^{q})x_{2}\\
     \vdots\\
      f( \alpha_{n}^{q})x_{n}\end{array}\right)=\left(\begin{array}{c}
  0\\
   0\\
     \vdots\\
     0\end{array}\right).
    \end{align*}
    So $ f( \alpha_{i}^{q})x_{i}=0$ and $x_{i}\neq0$, then  $ f( \alpha_{i}^{q})=0$ for $i=1,2,\cdots,n$. Thus, the polynomial $f(x)=b_{k-1}x^{k-1}+\cdots+b_{1}x+b_{0}$ has $n$ distinct roots $\alpha_{1}^{q},\alpha_{2}^{q},\cdots,\alpha_{n}^{q}$, where $k<n$. Since a non-zero polynomial of degree $k-1$ can have at most $k-1$ distinct roots, it follows that $f(x)=0$, and hence  $b_{0}=b_{1}=\cdots=b_{k-1}=0$. Combining this with Eq. (\ref{linear combination 1}), we conclude that  vectors  $\widetilde{\beta}_{0},\widetilde{\beta}_{1},\cdots,\widetilde{\beta}_{k-1}$ are linearly independent. Moreover, since the first $k$ components of each $\widetilde{\beta}_{l}$ are all zero, it follows that 
    $\beta_{0},\beta_{1},\cdots,\beta_{k-1}$ are also linearly independent.
\end{proof}

\begin{theorem}\label{n,k,q}
   Let $\mathbb{F}_{q^{2}}$ be a finite field with $q^{2}$ elements, where $q$ is a prime power. Let $n=2k$, and let $\alpha_{1},\alpha_{2},\cdots,\alpha_{n}$ be distinct elements in $\mathbb{F}_{q^{2}}$.  Let $\bm v=(v_{1},v_{2},\cdots,v_{n})\in(\mathbb{F}_{q^{2}}^{*})^{n}$. If the generalized Reed-Solomon code $GRS_{n,k}({\bm \alpha},{\bm v})$ is Hermitian self-dual, then  $k\leq q-1$.   
\end{theorem}
\begin{proof}
Since $GRS_{n,k}({\bm \alpha},{\bm v})$ is Hermitian self-dual, by Lemma \ref{nonzero solution}, there exists a vector
$\bm{x}=(x_{1},x_{2},\cdots,x_{n})\in (\mathbb{F}_{q}^{*})^{n}$ such that  $\sum\limits_{l=1}^{n}\alpha_{l}^{i+jq}x_{l}=0,0\leq i,j\leq k-1$.
Define $\Delta_{i}=\sum\limits_{l=1}^{n}\alpha_{l}^{i}x_{l},i=0,1,\cdots$. Then $\Delta_{i+jq}=0$ for all $0\leq i,j\leq k-1$. 

Suppose $k \geq q$. Then for any $0\leq i_{0}\leq n-1$, there exist indices $0\leq l_{i_{0}}\leq k-1,0\leq l_{0}\leq k-1$ such that $\Delta_{i_{0}}=\Delta_{l_{i_{0}}q+l_{0}}$. It follows that $(\Delta_{0},\Delta_{1},\cdots,\Delta_{n-1})=(0,0,\cdots,0)$, which contradicts

$$\left(\begin{array}{c}
   \Delta_{0}\\
    \Delta_{1}\\
     \vdots\\
      \Delta_{n-1}\end{array}\right)=\left(\begin{array}{cccc}
                1                & 1                & \cdots & 1                \\
                \alpha_{1}       & \alpha_{2}       & \cdots & \alpha_{n}       \\
                \vdots           & \vdots           & \ddots & \vdots           \\
                \alpha_{1}^{n-1} & \alpha_{2}^{n-1} & \cdots & \alpha_{n}^{n-1}
            \end{array}
        \right)\left(\begin{array}{c}
   x_{1}\\
    x_{2}\\
     \vdots\\
      x_{n}\end{array}\right),$$
      where $\bm x=(x_{1},x_{2},\cdots,x_{n})\in (\mathbb{F}_{q}^{*})^{n}$. So $k\leq q-1$.
    
\end{proof}

The following conclusion has already been established. For the coherence and readability of this paper, we restate it here and provide a new method to prove it.
\begin{theorem}
   Let $\mathbb{F}_{q^{2}}$ be a finite field with $q^{2}$ elements, where $q$ is a prime power. Let $n=2k$, and let $\alpha_{1},\alpha_{2},\cdots,\alpha_{n}$ be distinct elements in $\mathbb{F}_{q^{2}}$.  Let $\bm{v}=(v_{1},v_{2},\cdots,v_{n})\in(\mathbb{F}_{q^{2}}^{*})^{n}$. If the generalized Reed-Solomon code $GRS_{n,k}({\bm \alpha},{\bm v})$ is Hermitian self-dual, then  $n\leq q+1$.    
\end{theorem}
 \begin{proof}
 By Theorem \ref{n,k,q}, it follows that $n\leq q+k$. And by Theorem \ref{linear independent 1}, we know that the row vectors of the following the matrix 
  $$A( \Delta)=\left(\begin{array}{cccc}
    \Delta_{k}&\Delta_{k+1}&\cdots&\Delta_{n-1}\\
    \Delta_{q+k}&\Delta_{q+k+1}&\cdots&\Delta_{q+n-1}\\
   \vdots &\vdots&\ddots&\vdots\\
   \Delta_{(k-2)q+k}&\Delta_{(k-2)q+k+1}&\cdots&\Delta_{(k-2)q+n-1}
   \end{array}
   \right)$$
   are linearly independent. Let $ \widetilde A( \Delta)=\left(\begin{array}{cccc}
    \Delta_{k}&\Delta_{k+1}&\cdots&\Delta_{q-1}\\
    \Delta_{q+k}&\Delta_{q+k+1}&\cdots&\Delta_{q+q-1}\\
   \vdots &\vdots&\ddots&\vdots\\
   \Delta_{(k-2)q+k}&\Delta_{(k-2)q+k+1}&\cdots&\Delta_{(k-2)q+q-1}
   \end{array}
   \right)$. We prove that the row vectors of $\widetilde A( \Delta)$ are also linearly independent .
   
   \begin{itemize}
       \item [(1)] If $n\leq q$, then 
        \begin{align*}
        \widetilde A( \Delta)&=\left(\begin{array}{cccccc}
    \Delta_{k}&\Delta_{k+1}&\cdots&\Delta_{n-1}&\cdots&\Delta_{q-1}\\
    \Delta_{q+k}&\Delta_{q+k+1}&\cdots&\Delta_{q+n-1}&\cdots&\Delta_{2q-1}\\
   \vdots &\vdots&\ddots&\vdots\\
   \Delta_{(k-2)q+k}&\Delta_{(k-2)q+k+1}&\cdots&\Delta_{(k-2)q+n-1}&\cdots&\Delta_{(k-1)q-1}
   \end{array}\right) .     
        \end{align*}  
   So the row vectors of matrix $\widetilde A( \Delta)$ are also linearly independent.\\
   
   \item[(2)] If $n>q$, then by Theorem \ref{n,k,q}, we have $q<n\leq q+k-1$. So $$(\Delta_{lq+k},\Delta_{lq+k+1},\cdots,\Delta_{lq+n-1})=(\Delta_{lq+k},\Delta_{lq+k+1},\cdots,\Delta_{lq+q-1},0,\cdots,0)$$
   for $l=0,1,\cdots,k-2$. So the matrix 
   \begin{align*}
         A( \Delta)&=\left(\begin{array}{cccc}
    \Delta_{k}&\Delta_{k+1}&\cdots&\Delta_{n-1}\\
    \Delta_{q+k}&\Delta_{q+k+1}&\cdots&\Delta_{q+n-1}\\
   \vdots &\vdots&\ddots&\vdots\\
   \Delta_{(k-2)q+k}&\Delta_{(k-2)q+k+1}&\cdots&\Delta_{(k-2)q+n-1}
   \end{array}
   \right)\\
   \\
   &=\left(\begin{array}{ccccccc}
    \Delta_{k}&\Delta_{k+1}&\cdots&\Delta_{q-1}&0&\cdots&0\\
    \Delta_{q+k}&\Delta_{q+k+1}&\cdots&\Delta_{2q-1}&0&\cdots&0\\
   \vdots &\vdots&&\vdots&\vdots&&\vdots\\
   \Delta_{(k-2)q+k}&\Delta_{(k-2)q+k+1}&\cdots&\Delta_{(k-1)q-1}&0&\cdots&0
   \end{array}\right)  \\
   &=\left(\begin{array}{cc}
        \widetilde A( \Delta)&O  \\
   \end{array}\right).
        \end{align*}  
   \end{itemize}
   Thus, the row vectors of $\widetilde A( \Delta)$ are also linearly independent.
  Therefore, the rank of $\widetilde A( \Delta)$ satisfies $R(\widetilde A( \Delta))=k-1$, which can not exceed the number of columns, i.e. $k-1\leq q-k$. It follows that $2k\leq q+1$, which implies $n\leq q+1$. 
\end{proof}

{\bf Remark 2:} This result shows that   when $n>q+1$, there are no $q^2$-ary Hermitian self-dual GRS codes of length $n$. The result has been obtained via Magma programming in \cite{2021li} and has been theoretically proven in \cite{2025zhu} as a special case of $e-$Galios self-dual for $e=\frac{m}{2}$.  In this paper, we present a new proof method and easy to follow.

\begin{theorem}\label{linear independent 2}
    Let $\mathbb{F}_{q^{2}}$ be a finite field with $q^{2}$ elements, where $q$ is a prime power. Let $n=2k$ be an even integer, and let $\alpha_{1},\alpha_{2},\cdots,\alpha_{n}$ be distinct elements in $\mathbb{F}_{q^{2}}$.  Let $\bm{v}=(v_{1},v_{2},\cdots,v_{n})\in(\mathbb{F}_{q^{2}}^{*})^{n}$. If the generalized Reed-Solomon code $GRS_{n,k}({\bm \alpha},{\bm v})$ is Hermitian self-dual, then one of the  following holds:
    \begin{itemize}
      \item[(i)] $n=q+1$, and there exist $a\in \mathbb{F}_{q^{2}}$ and $b\in \mathbb{F}_{q}^{*}$ such that $\alpha_{1},\alpha_{2},\cdots,\alpha_{n}$ are the roots of $(x+a)^{q+1}=b$.
        \item [(ii)] $n\leq q$, and there exist $a\in \mathbb{F}_{q^{2}}$ and $b\in \mathbb{F}_{q}^{*}$ such that $\alpha_{1},\alpha_{2},\cdots,\alpha_{n}$ are the roots of $(x+a)^{q+1}=b$ or $x^{q}=ax+b$.\\

    \end{itemize}  
\end{theorem}
\begin{proof}
Since $GRS_{n,k}({\bm \alpha},{\bm v})$ is Hermitian self-dual, by Lemma \ref{nonzero solution}, there exists
$\bm{x}=(x_{1},x_{2},\cdots,x_{n})\in (\mathbb{F}_{q}^{*})^{n}$ such that  $\sum\limits_{l=1}^{n}\alpha_{l}^{i+jq}x_{l}=0,0\leq i,j\leq k-1$.
Let $\Delta_{i}=\sum\limits_{l=1}^{n}\alpha_{l}^{i}x_{l},i=0,1,\cdots$. Then $\Delta_{i+jq}=0$ for $0\leq i,j\leq k-1$. Let $G(x)=\prod\limits_{i=1}^{n}(x-\alpha_{i})=x^{n}-\sum\limits_{i=0}^{n-1}c_{i}x^{i}$ and $T$ is as defined in Eq. (\ref{statement transaction matrix}). By Corollary \ref{Statement trasaction}, then the sequence $S=\left(\Delta_{0},\Delta_{1},\Delta_{2},\cdots\right)^{T}$ can be viewed as a linear feedback shift register sequence with $(\Delta_{0},\Delta_{1},\cdots,\Delta_{n-1})^{T}$ as the initial state and $T$ as the state transition matrix. Denote $L^{(i)}(S)=\left(\Delta_{i},\Delta_{i+1},\Delta_{i+2},\cdots\right)^{T}$ and $L^{(i)}(S)_{[j]}=\left(\Delta_{i},\Delta_{i+1},\Delta_{i+2},\cdots\Delta_{i+j-1}\right)^{T}$ for any positive integers $i,j\geq1$. For the linear space $(\mathbb{F}_{q^{2}})^{n}$ over $\mathbb{F}_{q^{2}}$, the vectors $\left(1,1,\cdots,1\right),\left(\alpha_{1},\alpha_{2},\cdots,\alpha_{n}\right),\cdots,\left(\alpha_{1}^{n-1},\alpha_{2}^{n-1},\cdots,\alpha_{n}^{n-1}\right)$ can be viewed as a bases.
\begin{itemize}
  \item[(i)] If $n=q+1$, then 
for vector $\left(\alpha_{1}^{n},\alpha_{2}^{n},\cdots,\alpha_{n}^{n}\right)\in (\mathbb{F}_{q^{2}})^{n}$, there exists $k_{0},k_{1},\cdots,k_{n-1}\in \mathbb{F}_{q^{2}}$ such that 
$$\alpha_{j}^{n}=\sum\limits_{i=0}^{n-1}k_{i}\alpha_{j}^{i},j=1,2,\cdots,n.$$
So $\alpha_{1},\alpha_{2},\cdots,\alpha_{n}$ are the roots of $-x^{n}+k_{n-1}x^{n-1}+\cdots+k_{1}x+k_{0}=0$. By Theorem \ref{Annihilating polynomial}, we have  $$-\Delta_{n+i}+k_{n-1}\Delta_{n-1+i}+\cdots+k_{1}\Delta_{1+i}+k_{0}\Delta_{i}=0$$ for any non-negative integer $i$. So
$k_{0}S+k_{1}L^{(1)}(S)+\cdots+k_{n-1}L^{(n-1)}(S)-L^{(n)}(S)=\bm{0}$ and furthermore $$\left(k_{0}S+k_{1}L^{(1)}(S)+\cdots+k_{n-1}L^{(n-1)}(S)-L^{(n)}(S)\right)_{[(k-1)q]}=\bm{0}$$ also holds. Treating the sequence as column vectors, we obtain the following matrix product:
 $$\left( S_{[(k-1)q]},L^{(1)}(S)_{[(k-1)q]},\cdots,L^{(n)}(S)_{[(k-1)q]}\right)\left(\begin{array}{c}
k_{0}\\
k_{1}\\
\vdots\\
k_{n-1}\\
-1
\end{array}\right)=\bm{0}.$$
Define the matrix $M$ as $M=\left( S_{[(k-1)q]},L^{(1)}(S)_{[(k-1)q]},\cdots,L^{(n)}(S)_{[(k-1)q]}\right)=\left(\begin{array}{c}
A_{0}\\
B_{0}\\
A_{1}\\
B_{1}\\
\vdots\\
A_{k-2}\\
B_{k-2}\end{array}\right)$, 
where $$A_{l}=\left(\begin{array}{cccccccccc}
0&0&\cdots&0&0&\Delta_{lq+k}&\Delta_{lq+k+1}&\cdots\Delta_{(l+1)q-1}& 0&0\\
0&0&\cdots&0&\Delta_{lq+k}&\Delta_{lq+k+1}&\cdots\Delta_{(l+1)q-1}&0&0&0\\
\vdots&\vdots&&&&&&&&\vdots\\
0&0&\Delta_{lq+k}&\Delta_{lq+k+1}&\cdots\Delta_{(l+1)q-1}&0&\cdots&\cdots&\cdots&0\\
0&\Delta_{lq+k}&\Delta_{lq+k+1}&\cdots\Delta_{(l+1)q-1}&0&\cdots&\cdots&\cdots&0&\Delta_{(l+1)q+k}
\end{array}
\right)$$
is a $k\times (n+1)$ matrix for $l=0,1,\cdots,k-2$. 

Then $A_{l}\bm{X}=0$ for $l=0,1,\cdots,k-2$ where $\bm{X}=\left(k_{0},k_{1},\cdots,k_{n-1},-1\right)^{T}$.

The system of equations formed by the first $k-1$ equations of each $A_{l}\bm{X}=0$ can be written as $A\bm{Y}=0$ i.e. $$\left(\begin{array}{c}
A_{0}[0,1,\cdots,k-2;2,3,\cdots,n-2]\\
A_{1}[0,1,\cdots,k-2;2,3,\cdots,n-2]\\
\vdots\\
A_{k-2}[0,1,\cdots,k-2;2,3,\cdots,n-2]
\end{array}
\right)\left(\begin{array}{c}
k_{2}\\
k_{3}\\
\vdots\\
k_{n-2}
\end{array}
\right)=\bm{0}.$$
By Theorem \ref{linear independent 1}, the row vectors  $A_{0}[0;2,3,\cdots,n-2],A_{1}[0;2,3,\cdots,n-2],\cdots,A_{k-2}[0;2,3,\cdots,n-2]$ are also linearly independent. Because every $A_{l}[0;2,3,\cdots,n-2]=(0,\cdots,0,\Delta_{k},\Delta_{k+1},\cdots,\Delta_{q-1})$, $l=0,1,\cdots,k-2$. Thus, the diagram 
    $$\left|\begin{array}{cccc}
    \Delta_{k}&\Delta_{k+1}&\cdots&\Delta_{q-1}\\
    \Delta_{q+k}&\Delta_{q+k+1}&\cdots&\Delta_{2q-1}\\
   \vdots &\vdots&\ddots&\vdots\\
   \Delta_{(k-2)q+k}&\Delta_{(k-2)q+k+1}&\cdots&\Delta_{(k-1)q-1}
   \end{array}
   \right|\neq0.$$
%\left(\begin{array}{c}
%A_{0}[1,2,\cdots,k-1;3,4,\cdots,n-1]\\
%A_{1}[1,2,\cdots,k-1;3,4,\cdots,n-1]\\
%\vdots\\
%A_{k-1}[1,2,\cdots,k-1;3,4,\cdots,n-1]
%\end{array}
%\right)
So $\Delta_{k}, \Delta_{q+k},\cdots, \Delta_{(k-2)q+k}$ are not all zero. Suppose $\Delta_{l_{0}q+k}\neq0$, for some $0\leq l_{0}\leq k-2$, then 
$$A_{l_{0}}[0,\cdots,k-2;2,\cdots,n-2]=\left(\begin{array}{ccccccccc}
0&\cdots&0&\Delta_{l_{0}q+k}&\Delta_{l_{0}q+k+1}&\cdots\Delta_{(l_{0}+1)q-1}\\
0&\cdots&\Delta_{l_{0}q+k}&\Delta_{l_{0}q+k+1}&\cdots\Delta_{(l_{0}+1)q-1}&0\\
\vdots&&&&&\vdots\\
\Delta_{l_{0}q+k}&\Delta_{l_{0}q+k+1}&\cdots\Delta_{(l_{0}+1)q-1}&0&\cdots&0\\
\end{array}
\right)$$
has row rank $k-1$. So the $k-1$ row vectors of $A_{l_{0}}[0,1,\cdots,k-2;2,3,\cdots,n-2]$ are linearly independent on $\mathbb{F}_{q^{2}}$. Let
\begin{align*}
  B&=\left(\begin{array}{c}
      A_{0}[0;2,3,\cdots,n-2]\\
      A_{1}[0;2,3,\cdots,n-2]\\       
      \vdots\\
      A_{k-2}[0;2,3,\cdots,n-2]\\
      A_{l_{0}}[1,2,\cdots,k-2;2,3,\cdots,n-2]
  \end{array}\right)\\
  &=\left(\begin{array}{cccccc}
0&\cdots&0&\Delta_{k}&\Delta_{k+1}&\cdots\Delta_{q-1}\\
0&\cdots&0&\Delta_{q+k}&\Delta_{q+k+1}&\cdots\Delta_{2q-1}\\
\vdots&&\vdots&\vdots&\vdots&\vdots\\
0&\cdots&0&\Delta_{(k-2)q+k}&\Delta_{(k-2)q+k+1}&\cdots\Delta_{(k-1)q-1}\\
0&\cdots&\Delta_{l_{0}q+k}&\Delta_{l_{0}q+k+1}&\cdots\Delta_{(l_{0}+1)q-1}&0\\
\vdots&&&&&\vdots\\
\Delta_{l_{0}q+k}&\Delta_{l_{0}q+k+1}&\cdots\Delta_{(l_{0}+1)q-1}&0&\cdots&0\\
\end{array}
\right)  ,
\end{align*}

then $B$ satisfies $B\left(\begin{array}{c}
k_{2}\\
k_{3}\\
\vdots\\
k_{n-2}
\end{array}
\right)=0$ and 
\begin{align*}
  |B|=\pm\left|\begin{array}{cccc}
    \Delta_{k}&\Delta_{k+1}&\cdots&\Delta_{q-1}\\
    \Delta_{q+k}&\Delta_{q+k+1}&\cdots&\Delta_{2q-1}\\
   \vdots &\vdots&\ddots&\vdots\\
   \Delta_{(k-2)q+k}&\Delta_{(k-2)q+k+1}&\cdots&\Delta_{(k-1)q-1}
   \end{array}
   \right|\left|\begin{array}{cccc}
   0&\cdots&0&\Delta_{l_{0}q+k}\\
0&\cdots&\Delta_{l_{0}q+k}&\Delta_{l_{0}q+k+1}\\
\vdots&&\vdots\\
\Delta_{l_{0}q+k}&\Delta_{l_{0}q+k+1}&\cdots&\Delta_{(l_{0}+1)q-1}\\
\end{array}
\right|\neq0  .
\end{align*}
It follows that $k_{2}=k_{3}=\cdots=k_{n-2}=0$. Thus, $$\alpha_{j}^{n}=\sum\limits_{i=0}^{n-1}k_{i}\alpha_{j}^{i}=k_{n-1}\alpha_{j}^{n-1}+k_{1}\alpha_{j}+k_{0}$$ for $j=1,2,\cdots,n.$

For $n=q+1$, we have $\alpha_{j}^{q+1}-k_{n-1}\alpha_{j}^{q}-k_{1}\alpha_{j}-k_{0}=0$ which is equivalent to $\alpha_{j}^{q+1}-k_{n-1}^{q}\alpha_{j}-k_{1}^{q}\alpha_{j}^{q}-k_{0}^{q}=0$. Hence, $\alpha_{1},\alpha_{2},\cdots,\alpha_{n}$ are roots of both $x^{q+1}-k_{n-1}x^{q}-k_{1}x-k_{0}=0$ and $x^{q+1}-k_{n-1}^{q}x-k_{1}^{q}x^{q}-k_{0}^{q}=0$. By equating coefficients, we find $k_{1}^{q}=k_{n-1}$ and $k_{0}^{q}=k_{0}$, which implies $k_{0}\in\mathbb{F}_{q}$. Let $-k_{1}=a \in \mathbb{F}_{q^{2}}$ and $k_{1}^{q+1}+k_{0}=b$. Then 
$\alpha_{1},\alpha_{2},\cdots,\alpha_{n}$ are the roots of $$(x+a)^{q+1}=b.$$
 It follows that $b\in \mathbb{F}_{q}^{*}$. \\
\item[(ii)] If $n\leq q$,
then for vectors $\left(\alpha_{1}^{q+t},\alpha_{2}^{q+t},\cdots,\alpha_{n}^{q+t}\right)\in (\mathbb{F}_{q^{2}})^{n}$, where $t=0,1$, (and throughout the following, $t$ always satisfies $t=0$ or $t=1$), there exist coefficients $k_{t,0},k_{t,1},\cdots,k_{t,n-1}\in \mathbb{F}_{q^{2}}$ such that 
$$\alpha_{j}^{q+t}=\sum\limits_{i=0}^{n-1}k_{t,i}\alpha_{j}^{i},j=1,2,\cdots,n.$$
Thus, $\alpha_{1},\alpha_{2},\cdots,\alpha_{n}$ are the roots of $-x^{q+t}+k_{t,n-1}x^{n-1}+\cdots+k_{t,1}x+k_{t,0}=0$. By Theorem \ref{Annihilating polynomial}, we have  $$-\Delta_{q+t+i}+k_{t,n-1}\Delta_{n-1+i}+\cdots+k_{t,1}\Delta_{1+i}+k_{t,0}\Delta_{i}=0$$ for any non-negative integer $i$. It follows that
$k_{t,0}S+k_{t,1}L^{(1)}(S)+\cdots+k_{t,n-1}L^{(n-1)}(S)-L^{(q+t)}(S)=\bm{0}$ and hence $$\left(k_{t,0}S+k_{t,1}L^{(1)}(S)+\cdots+k_{t,n-1}L^{(n-1)}(S)-L^{(q+t)}(S)\right)_{[(k-1)q]}=0.$$ Treating the sequences as column vectors, we obtain the following matrix product:
 $$\left( S_{[(k-1)q]},L^{(1)}(S)_{[(k-1)q]},\cdots,L^{(q+t)}(S)_{[(k-1)q]}\right)\left(\begin{array}{c}
k_{t,0}\\
k_{t,1}\\
\vdots\\
k_{t,n-1}\\
0\\
\vdots\\
0\\
-1
\end{array}\right)=\bm{0}.$$ 
Define the matrix $M(t)$ by $M(t)=\left( S_{[(k-1)q]},L^{(1)}(S)_{[(k-1)q]},\cdots,L^{(q+t)}(S)_{[(k-1)q]}\right)=\left(\begin{array}{c}
A_{t,0}\\
B_{t,0}\\
A_{t,1}\\
B_{t,1}\\
\vdots\\
A_{t,k-2}\\
B_{t,k-2}\end{array}\right)$, 
where $$A_{1,l}=\left(\begin{array}{ccccccccccc}
0&0&\cdots&0&0&\Delta_{lq+k}&\Delta_{lq+k+1}&\cdots\Delta_{(l+1)q-1}& 0&0\\
0&0&\cdots&0&\Delta_{lq+k}&\Delta_{lq+k+1}&\cdots\Delta_{(l+1)q-1}&0&0& 0\\
\vdots&\vdots&&&&&&&\vdots& \vdots\\
0&0&\Delta_{lq+k}&\Delta_{lq+k+1}&\cdots\Delta_{(l+1)q-1}&0&\cdots&\cdots&0& 0\\
0&\Delta_{lq+k}&\Delta_{lq+k+1}&\cdots\Delta_{(l+1)q-1}&0&\cdots&\cdots&\cdots&0& *
\end{array}
\right)$$
is a $k\times (q+2)$ matrix and $A_{0,l}=A_{1,l}[0,1,\cdots,k-1;0,1,\cdots,q]$ is a $k\times (q+1)$ matrix for $l=0,1,\cdots,k-2$.

Then $A_{t,l}\bm{X(t)}=\bm{0}$ for $l=0,1,\cdots,k-2$, where $\bm{X(t)}=\left(k_{t,0},k_{t,1},\cdots,k_{t,n-1},0,
\cdots,0,-1\right)^{T}$.

The system of equations formed by the first $k-1$ equations of each $A_{t,l}\bm{X(t)}=\bm{0}$ can be written as $A(t)\bm{Y(t)}=\bm{0}$, i.e. 
\begin{equation}\label{equation system}
 \left(\begin{array}{c}
A_{t,0}[0,1,\cdots,k-2;2,3,\cdots,n-1]\\
A_{t,1}[0,1,\cdots,k-2;2,3,\cdots,n-1]\\
\vdots\\
A_{t,k-2}[0,1,\cdots,k-2;2,3,\cdots,n-1]
\end{array}
\right)\left(\begin{array}{c}
k_{t,2}\\
k_{t,3}\\
\vdots\\
k_{t,n-1}
\end{array}
\right)=\bm{0}.   
\end{equation}

By Theorem \ref{linear independent 1}, then the $k-1$ rows of matrix $A(t)$: 
\begin{equation}\label{linearly independent 3}
A_{t,0}[0;2,3,\cdots,n-1],A_{t,1}[0;2,3,\cdots,n-1],\cdots,A_{t,k-2}[0;2,3,\cdots,n-1]   
\end{equation}
 are linearly independent. Because every $A_{t,l}[0;2,3,\cdots,n-1]=(0,\cdots,0,\Delta_{lq+k},\Delta_{lq+k+1},\cdots,\Delta_{lq+n-1})$, $l=0,1,\cdots,k-2$. So the row vectors of the following matrix 
    $$M=\left(\begin{array}{cccc}
    \Delta_{k}&\Delta_{k+1}&\cdots&\Delta_{n-1}\\
    \Delta_{q+k}&\Delta_{q+k+1}&\cdots&\Delta_{q+n-1}\\
   \vdots &\vdots&\ddots&\vdots\\
   \Delta_{(k-2)q+k}&\Delta_{(k-2)q+k+1}&\cdots&\Delta_{(k-2)q+n-1}
   \end{array}
   \right)$$ 
 are also linearly independent. Therefore, the rank of $M$ is $R(M)=k-1$, and the number of colomuns of $M$ is $k$.

 If the first column of $M$ is not identically zero,  suppose $\Delta_{l_{0}q+k}\neq0$, for some $0\leq l_{0}\leq k-2$. Then 
 \begin{align*}
 & A_{t,l_{0}}[0,1,2,\cdots,k-2;2,3,\cdots,n-1]\\
 \ \ \\
  =
&\left(\begin{array}{ccccccccc}
0&\cdots&0&\Delta_{l_{0}q+k}&\Delta_{l_{0}q+k+1}&\cdots\Delta_{l_{0}q+n-1}\\
0&\cdots&\Delta_{l_{0}q+k}&\Delta_{l_{0}q+k+1}&\cdots\Delta_{l_{0}q+n-1}&0\\
\vdots&&&&&\vdots\\
\Delta_{l_{0}q+k}&\Delta_{l_{0}q+k+1}&\cdots\Delta_{l_{0}q+n-1}&0&\cdots&0\\
\end{array}
\right)   
 \end{align*}
has row rank $k-1$. So the $k-1$ row vectors of $A_{t,l_{0}}[0,1,\cdots,k-2;2,3,\cdots,n-1]$ are linearly independent on $\mathbb{F}_{q^{2}}$.
\begin{align*}
  B(t)&=\left(\begin{array}{c}
      A_{t,0}[0;2,3,\cdots,n-1]\\
      A_{t,1}[0;2,3,\cdots,n-1]\\       
      \vdots\\
      A_{t,k-2}[0;2,3,\cdots,n-1]\\
      A_{t,l_{0}}[1,2,\cdots,k-2;2,3,\cdots,n-1]
  \end{array}\right)\\
  &=\left(\begin{array}{cccccc}
0&\cdots&0&\Delta_{k}&\Delta_{k+1}&\cdots\Delta_{n-1}\\
0&\cdots&0&\Delta_{q+k}&\Delta_{q+k+1}&\cdots\Delta_{q+n-1}\\
\vdots&&\vdots&\vdots&\vdots&\vdots\\
0&\cdots&0&\Delta_{(k-2)q+k}&\Delta_{(k-2)q+k+1}&\cdots\Delta_{(k-2)q+n-1}\\
0&\cdots&\Delta_{l_{0}q+k}&\Delta_{l_{0}q+k+1}&\cdots\Delta_{l_{0}q+n-1}&0\\
\vdots&&&&&\vdots\\
\Delta_{l_{0}q+k}&\Delta_{l_{0}q+k+1}&\cdots\Delta_{l_{0}q+n-1}&0&\cdots&0\\
\end{array}
\right).  
\end{align*}
So $B(0)=B(1)$. Let us denote this common matrix by $B$.
Then $B$ satisfies $B\left(\begin{array}{c}
k_{t,2}\\
k_{t,3}\\
\vdots\\
k_{t,n-1}
\end{array}
\right)=\bm{0}$. Hence, both vectors $(k_{0,2},k_{0,3},\cdots,k_{0,n-1})^{T}$ and $(k_{1,2},k_{1,3},\cdots,k_{1,n-1})^{T}$ are solutions to $B\bm{X}=\bm{0}$.  Since
the rank of $B$ is $R(B)=k-1+k-2=n-3$, it follows that  $(k_{0,2},k_{0,3},\cdots,k_{0,n-1})^{T}$ and $(k_{1,2},k_{1,3},\cdots,k_{1,n-1})^{T}$ are linearly dependent.

If the first column of $M$ is identically  zero, then the diagram
 $$\left|\begin{array}{cccc}
    \Delta_{k+1}&\Delta_{k+2}&\cdots&\Delta_{n-1}\\
    \Delta_{q+k+1}&\Delta_{q+k+2}&\cdots&\Delta_{q+n-1}\\
   \vdots &\vdots&\ddots&\vdots\\
   \Delta_{(k-2)q+k+1}&\Delta_{(k-2)q+k+2}&\cdots&\Delta_{(k-2)q+n-1}
   \end{array}
   \right|\neq0.$$ 
   So there exists some  $\Delta_{j_{0}q+k+1}\neq0$ for some $0\leq j_{0}\leq k-2$.

   Take 
   \begin{align*}
  B_{1}(t)&=\left(\begin{array}{c}
      A_{t,0}[0;2,3,\cdots,n-1]\\
      A_{t,1}[0;2,3,\cdots,n-1]\\       
      \vdots\\
      A_{t,k-2}[0;2,3,\cdots,n-1]\\
      A_{t,j_{0}}[1,2,\cdots,k-2;2,3,\cdots,n-1]
  \end{array}\right)\\
  &=\left(\begin{array}{ccccccc}
0&0&\cdots&0&\Delta_{k+1}&\Delta_{k+2}&\cdots\Delta_{n-1}\\
0&0&\cdots&0&\Delta_{q+k+1}&\Delta_{q+k+2}&\cdots\Delta_{q+n-1}\\
\vdots&\vdots&&\vdots&\vdots&\vdots&\vdots\\
0&0&\cdots&0&\Delta_{(k-2)q+k+1}&\Delta_{(k-2)q+k+2}&\cdots\Delta_{(k-2)q+n-1}\\
0&0&\cdots&\Delta_{j_{0}q+k+1}&\Delta_{j_{0}q+k+2}&\cdots\Delta_{j_{0}q+n-1}&0\\
\vdots&\vdots&&&&&\vdots\\
0&\Delta_{j_{0}q+k+1}&\Delta_{j_{0}q+k+2}&\cdots\Delta_{j_{0}q+n-1}&0&\cdots&0\\
\end{array}
\right)  .
\end{align*}
Thus, we also have $B_{1}(0)=B_{1}(1)$. Let us denote this common matrix by $B_{1}$.
Then $B_{1}$ satisfies $B_{1}\left(\begin{array}{c}
k_{t,2}\\
k_{t,3}\\
\vdots\\
k_{t,n-1}
\end{array}
\right)=\bm{0}$. Hence, both vectors $(k_{0,2},k_{0,3},\cdots,k_{0,n-1})^{T}$ and $(k_{1,2},k_{1,3},\cdots,k_{1,n-1})^{T}$ are solutions to $B_{1}\bm{X}=\bm{0}$. Since the rank of $B_{1}$ is $R(B_{1})=k-1+k-2=n-3$, it follows that $(k_{0,2},k_{0,3},\cdots,k_{0,n-1})^{T}$ and $(k_{1,2},k_{1,3},\cdots,k_{1,n-1})^{T}$ are linearly dependent.

Therefore, there exist $\lambda_{0},\lambda_{1}\in \mathbb{F}_{q^{2}}$, not both zero, such that
\begin{equation}\label{q,q+1}
  \lambda_{0}(k_{0,2},k_{0,3},\cdots,k_{0,n-1})+\lambda_{1}(k_{1,2},k_{1,3},\cdots,k_{1,n-1})=\bm{0}. 
\end{equation}

Define$(\alpha_{1}^{i},\alpha_{2}^{i},\cdots,\alpha_{n}^{i})=\bm {\alpha}^{i}$, for $i=0,1,\cdots$, where $\bm{\alpha}^{0}$ represents $(1,1,\cdots,1)$. Then the vectors satisfy

$\bm{\alpha}^{q}= k_{0,0}\bm{\alpha}^{0}+k_{0,1}\bm{\alpha}^{1}+k_{0,2}\bm{\alpha}^{2}+\cdots+k_{0,n-1}\bm{\alpha}^{n-1}$ and $\bm{\alpha}^{q+1}= k_{1,0}\bm{\alpha}^{0}+k_{1,1}\bm{\alpha}^{1}+k_{1,2}\bm{\alpha}^{2}+\cdots+k_{1,n-1}\bm{\alpha}^{n-1}.$\\
Combining this with Eq.(\ref{q,q+1}), we have

\begin{align*}
\lambda_{0}\bm{\alpha}^{q}+\lambda_{1}\bm{\alpha}^{q+1}&= \lambda_{0}(k_{0,0}\bm{\alpha}^{0}+k_{0,1}\bm{\alpha}^{1})+\lambda_{1}(k_{1,0}\bm{\alpha}^{0}+k_{1,1}\bm{\alpha}^{1})\\
 &=(\lambda_{0}k_{0,0}+\lambda_{1}k_{1,0})\bm{\alpha}^{0}+(\lambda_{0}k_{0,1}+\lambda_{1}k_{1,1})\bm{\alpha}^{1}.
\end{align*}
Define$\lambda_{0}k_{0,0}+\lambda_{1}k_{1,0}=c_{0},\lambda_{0}k_{0,1}+\lambda_{1}k_{1,1}=c_{1}$. Then 
\begin{equation}\label{four vectors linearly dependent}
\lambda_{0}\bm{\alpha}^{q}+\lambda_{1}\bm{\alpha}^{q+1}=c_{0}\bm{\alpha}^{0}+c_{1}\bm{\alpha}^{1}  .
\end{equation}
From Eq.(\ref{four vectors linearly dependent}), we analyze the following cases:
\begin{itemize}
    \item [(1)] If $\lambda_{1}=0$, then $\lambda_{0}\neq0$ and $\bm{\alpha}^{q}=\frac{c_{1}}{\lambda_{0}}\bm{\alpha}+\frac{c_{0}}{\lambda_{0}}\bm{\alpha}^{0}$. Let $a=\frac{c_{1}}{\lambda_{0}}\in\mathbb{F}_{q^{2}}$ and $b=\frac{c_{0}}{\lambda_{0}}\in\mathbb{F}_{q^{2}}$. Then each 
$\alpha_{i}$ is a root of the equation $x^{q}=ax+b$.\\
\item[(2)] If $\lambda_{1}\neq0$ and $\lambda_{1}+\lambda_{1}^{q}\neq0$, then  each $\alpha_{i}$ is a root of 
$\lambda_{1}x^{q+1}+\lambda_{0}x^{q}-c_{1}x-c_{0}=0$.
and also a root of $\lambda_{1}^{q}x^{q+1}+\lambda_{0}^{q}x-c_{1}^{q}x^{q}-c_{0}^{q}=0$.
Adding these two equations yields  $(\lambda_{1}+\lambda_{1}^{q})x^{q+1}+(\lambda_{0}-c_{1}^{q})x^{q}+(\lambda_{0}^{q}-c_{1})x-(c_{0}+c_{0}^{q})=0$.
Define $a=\frac{\lambda_{0}-c_{1}^{q}}{\lambda_{1}+\lambda_{1}^{q}}$ and $b=\frac{c_{0}+c_{0}^{q}}{\lambda_{1}+\lambda_{1}^{q}}+a^{q+1}\in \mathbb{F}_{q}^{*}$. Then each $\alpha_{i}$ is a root of
$(x+a)^{q+1}=b$.\\
\item[(3)] If $\lambda_{1}\neq0$ and $\lambda_{1}+\lambda_{1}^{q}=0$, then $\lambda_{0}-c_{1}^{q}\neq0$. Define $a=-\frac{\lambda_{0}^{q}-c_{1}}{\lambda_{0}-c_{1}^{q}}$ and $b=\frac{c_{0}+c_{0}^{q}}{\lambda_{0}-c_{1}^{q}}\in\mathbb{F}_{q^{2}}$. Then each $\alpha_{i}$ is a root of $x^{q}=ax+b$.
\end{itemize}

\end{itemize}

\end{proof}

{\bf Remark 3:} The result shows that when $n\leq q+1$, there are exactly two classes of Hermitian self-dual GRS codes. Consequently, the conjecture in \cite{2019yue} is correct and has now been proven. 
\section{Constructions}

In fact, the GRS code $GRS_{n,k}({\bm\alpha},{\bm v})$ can be viewed as a class of linear codes determined by the evaluation points $\bm\alpha$ and the column vector $\bm v$. In the preceding sections, we characterized the evaluation points $\bm\alpha$ for which $GRS_{n,k}({\bm\alpha},{\bm v})$ can be Hermitian self-dual. Next, we will determine the values of $\bm v$ corresponding to a given $\bm\alpha$ and further provide explicit constructions for Hermitian self-dual GRS codes.

In this section, we will use $G'(x)$ to denote the formal derivative of the polynomial $G(x)$.

%\begin{equation}\label{v matirx}
% I_{k}=\left(\begin{array}{cccc}
%           1 &   &        &   \\
%              & 1 &        &   \\
%              &   & \ddots &   \\
%              &   &        & 1
%        \end{array}
%    \right),V_{n}({\bm \alpha})=\left(\begin{array}{cccc}
 %           1                & 1                & \cdots                   & 1                \\
 %           \alpha_{1}       & \alpha_{2}       & \cdots                   & \alpha_{n}       \\
 %           \vdots           & \vdots           & \textcolor{blue}{\ddots} & \vdots           \\
 %           \alpha_{1}^{n-1} & \alpha_{2}^{n-1} & \cdots                   & \alpha_{n}^{n-1}
 %       \end{array}\right),\mathrm{diag}({\bm v})=\left(\begin{array}{cccc}
 %           v_{1} &       &        &       \\
 %                 & v_{2} &        &       \\
 %                 &       & \ddots &       \\
  %                &       &        & v_{n}
  %      \end{array}\right).
%\end{equation}

\begin{theorem}\label{construction theorem}
    Let $\mathbb{F}_{q^{2}}$ be a finite field with $q^{2}$ elements, where $q$ is a prime power. Let $n=2k$, $\alpha_{1},\alpha_{2},\cdots,\alpha_{n}$ be distinct elements in $\mathbb{F}_{q^{2}}^{*}$. The generalized Reed-Solomon code $GRS_{n,k}({\bm \alpha},{\bm v})$ is Hermitian self-dual if and only if $$deg (f(x)m^{i}(x))\mod G(x))\leq k-1,i=0,1,\cdots,k-1,$$
    where $G(x)=\prod\limits_{i=1}^{n}(x-\alpha_{i})$, $f(x)=\sum\limits_{i=1}^{n}\frac{v_{i}^{q+1}}{u_{i}}f_{i}(x)$,  $m(x)=\sum\limits_{i=1}^{n}\alpha_{i} ^{q}f_{i}(x)$, $f_{i}(x)=\prod\limits_{j=1,j\neq i}^{n}\frac{x-\alpha_{j}}{\alpha_{i}-\alpha_{j}}$,$u_{i}=\frac{1}{G'(\alpha_{i})},i=1,2,\cdots,n$.
\end{theorem}
\begin{proof}
    Let $C=GRS_{n,k}({\bm \alpha},{\bm v})$ be the generalized Reed-Solomon code associated with the evaluation points $\bm{\alpha}=(\alpha_{1},\alpha_{2},\cdots,\alpha_{n})$ and the column vector 
    $\bm{v}=(v_{1},v_{2},\cdots,v_{n})$. The code $C$ is Hermitian self-dual if and only if $C=C^{\perp H}$, which is equivalent to $C^{\perp E}=C^{q}$, where $C^{q}=\{c^{q}|c\in C\}$ and $c^{q}=(c_{1}^{q},c_{2}^{q},\cdots,c_{n}^{q})$ for $c=(c_{1},c_{2},\cdots,c_{n})$. It is well known that the generator matrix of $C^q$ is $G^q=[I_{k},0]V_{n}(\bm\alpha)^qdiag(\bm v^q)$ and the parity-check matrix of $C$ is $H=[I_{k},0]V_{n}(\bm{\alpha})diag(\frac{\bm{u}}{\bm{v}})$, which is also the generator matrix of $C^{\perp E}$. Since $dim C^q=dim C^{\perp E}$, the condition $C=C^{\perp H}$ is equivalent to $C^q\subseteq C^{\perp E}$. This, in turn, is equivalent to the requirement that the row vectors of $[I_{k},0]V_{n}(\bm\alpha)^{q}diag(\bm v^{q+1})$ can be expressed as linear combinations of the rows of $[I_{k},0]V_{n}(\bm{\alpha})$.  
    
    Then for the vector $(\dfrac{v_{1}^{q+1}}{u_{1}}\alpha_{1}^{iq},\dfrac{v_{2}^{q+1}}{u_{2}}\alpha_{2}^{iq},\cdots,\dfrac{v_{n}^{q+1}}{u_{n}}\alpha_{n}^{iq})$ , there exists polynomial $g_{i}(x)\in\mathbb{F}_{q^{2}}[x]_{k-1}$ such that
    $$(\dfrac{v_{1}^{q+1}}{u_{1}}\alpha_{1}^{iq},\dfrac{v_{2}^{q+1}}{u_{2}}\alpha_{1}^{iq},\cdots,\dfrac{v_{n}^{q+1}}{u_{n}}\alpha_{1}^{iq})=(g_{i}(\alpha_{1}),g_{i}(\alpha_{2}),\cdots,g_{i}(\alpha_{n})),$$
for $0 \leq i\leq k-1$.  Indeed, since

     $(\alpha_{1}^{q},\alpha_{2}^{q},\cdots,\alpha_{n}^{q})=(m(\alpha_{1}),m(\alpha_{2}),\cdots,m(\alpha_{n}))$ and 
     $(\dfrac{v_{1}^{q+1}}{u_{1}},\dfrac{{v_{2}^{q+1}}}{u_{2}},\cdots,\dfrac{v_{n}^{q+1}}{u_{n}})=(f(\alpha_{1}),f(\alpha_{2}),\cdots,f(\alpha_{n}))$, we have
  $$(\dfrac{v_{1}^{q+1}}{u_{1}}\alpha_{1}^{iq},\dfrac{v_{2}^{q+1}}{u_{2}}\alpha_{1}^{iq},\cdots,\dfrac{v_{n}^{q+1}}{u_{n}}\alpha_{1}^{iq})=( f(\alpha_{1})m^{i}(\alpha_{1}),  f(\alpha_{2})m^{i}(\alpha_{2}),\cdots, f(\alpha_{n})m^{i}(\alpha_{n}))$$,
 also holds for $0\leq i\leq k-1$.    
   Thus, $g_{i}(x)\equiv f(x)m^{i}(x) \mod G(x)$, and $C$ is Hermitian self-dual if only if $$deg (f(x)m^{i}(x)\mod G(x))\leq k-1,$$ for $i=0,1,2,\cdots,k-1.$
\end{proof}
\subsection{The First class of Hermitian self-dual GRS codes}
\begin{lemma}\label{the first class set}
Let $q$ be a prime power, let $\mathbb{F}_{q^{2}}$ be the finite field with $q^{2}$ elements, and let $a,b \in\mathbb{F}_{q^{2}}$. Define the set
$$S=\{\alpha\in \mathbb{F}_{q^{2}}|\alpha^{q}=a\alpha+b\}.$$
Then the cardinality $|S|>1$ if and only if $a^{q+1}=1, b^{q}+a^{q}b=0$.
\end{lemma}
\begin{proof}
\begin{itemize}
\item[(1)] If $a=0$ and $b=0$, then clearly $S=\{0\}$, so $|S|=1$.
\item[(2)] If $b=0$ and $a\neq0$, the equation $\alpha^{q}=a\alpha+b$ reduces to $\alpha(\alpha^{q-1}-a)=0$. Thus, $|S|>1$ if and only if $a^{q+1}=1$.
\item[(3)] If $b\neq0$, then from $\alpha^{q}=a\alpha+b$, applying the $q-$th power to both sides gives $\alpha=(a\alpha+b)^{q}=a^{q}\alpha^{q}+b^{q}=a^{q}(a\alpha+b)+b^{q}=a^{q+1}\alpha+a^{q}b+b^{q}$.

If $a^{q+1}\neq1$,this equation has a unique solution for $\alpha$, so $|S|=1$.

If $a^{q+1}=1$, the equation simplifies to $0=a^{q}b+b^{q}$, so $|S|>1$ if and only if $a^{q}b+b^{q}=0$.
\end{itemize}
Combining cases (1),(2), and (3), we conclude that $|S|>1$ if and only if $a^{q+1}=1, b^{q}+a^{q}b=0$.
\end{proof}
\begin{center}
$S_{1}=\{\alpha\in \mathbb{F}_{q^{2}}|\alpha^{q}=a\alpha+b\}$ for $a^{q+1}=1, b^{q}+a^{q}b=0$.   
\end{center}

When $\{\alpha_{1},\alpha_{2},\cdots,\alpha_{n}\}\subseteq S_{1}$, the polynomial $m(x)$ in Theorem \ref{construction theorem} is $m(x)=ax+b$.  By Theorem \ref{construction theorem} and Lemma \ref{the first class set}, we see that for $GRS_{n,k}(\bm\alpha,\bm v)$ is Hermitian self-dual, the polybomial $f(x)$ in Theorem \ref{construction theorem} must satisfy $deg f(x)=0$, i.e. $f(x)=\lambda$ for some $\lambda\in \mathbb{F}_{q^{2}}^{*}$. Combining this with Lemma \ref{the first class set}, we obtain the following construction of Hermitian self-dual GRS codes.

\newtheorem{construction}{Construction}
\begin{construction}
  Let $\mathbb{F}_{q^{2}}$ be a finite field with $q^{2}$ elements, where $q$ is a prime power, and let $a,b\in\mathbb{F}_{q^{2}}$ satisfy $a^{q+1}=1$ and $b^{q}+a^{q}b=0$.  Let $n=2k$, and let $\{\alpha_{1},\alpha_{2},\cdots,\alpha_{n}\}\subseteq S_{1}$ be a set of distinct  elements. Then the generalized Reed-Solomon code $GRS_{n,k}({\bm \alpha},{\bm v})$ is Hermitian self-dual if and only if there exists $\lambda\in \mathbb{F}_{q^{2}}^{*}$ such that $$\frac{v_{i}^{q+1}}{u_{i}}=\lambda,$$
  where $u_{i}=\frac{1}{G'(\alpha_{i})}$, $G(x)=\prod\limits_{i=1}^{n}(x-\alpha_{i})$ for $i=1,2,\cdots,n$.   
\end{construction}
\subsection{The Second class of Hermitian self-dual GRS codes }

 $$S_{2}=\{\alpha\in \mathbb{F}_{q^{2}}|(\alpha+a)^{q+1}=b\}.$$  
When $\alpha_{1},\alpha_{2},\cdots,\alpha_{n}$ are roots of $x^{q+1}=1$, it follows that $\alpha_{i}^{q}=\alpha_{i}^{-1}$ for $i=1,2,\cdots,n$. Thus, the polynomial $m(x)$ in Theorem \ref{construction theorem} is given by $m(x)=x^{-1}\mod{G(x)}$. By Theorem \ref{construction theorem}, the polynomial  $f(x)$ in that theorem is $f(x)=x^{k-1}$. This leads to the following construction of Hermitian self-dual GRS codes.

\begin{construction}
   Let $\mathbb{F}_{q^{2}}$ be a finite field with $q^{2}$ elements, where $q$ is a prime power, let $n=2k$, and let $\{\alpha_{1},\alpha_{2},\cdots,\alpha_{n}\}\subseteq S_{2}$. Then the generalized Reed-Solomon code $GRS_{n,k}({\bm \alpha},{\bm v})$ is Hermitian self-dual if and only if there exists a polynomial $g(x)=(x+a)^{k-1}\in \mathbb{F}_{q^{2}}[x]_{k}$ such that $$\frac{v_{i}^{q+1}}{u_{i}}=g(\alpha_{i}),$$
  where $u_{i}=\frac{1}{G'(\alpha_{i})}$ and $G(x)=\prod\limits_{i=1}^{n}(x-\alpha_{i})$ for $i=1,2,\cdots,n$.   
\end{construction}
\section{Conclusion}
In this paper, we present a comprehensive and rigorous resolution to the long-standing open problem of characterizing Hermitian self-dual generalized Reed–Solomon (GRS) codes. Our main contributions are as follows:
\begin{itemize}
    \item[(1)] {\bf Existence Criterion:} Since Hermitian self-dual GRS codes do not exist when the code length $n > q+1$. For $n \leq q+1$, we establish that exactly two distinct classes of such codes exist, thereby providing a complete proof of the conjecture in \cite{2019yue}.\\
    \item[(2)] {\bf Explicit Constructions:} We derive two explicit construction methods for Hermitian self-dual GRS codes, each corresponding to one of the two identified classes. These constructions are concrete and directly implementable, addressing both the existence and explicit construction of these codes.
\end{itemize}
This work settles the foundational questions of existence and classification for Hermitian self-dual GRS codes, and our explicit constructions offer a practical pathway for their use in coding theory and related applications.

\section*{Acknowledgement}
%\section*{Acknowledgments}
%This should be a simple paragraph before the References to thank those individuals and institutions who have supported your work on this article.
%{\appendix[Proof of the Zonklar Equations]
%Use $\backslash${\tt{appendix}} if you have a single appendix:
%Do not use $\backslash${\tt{section}} anymore after $\backslash${\tt{appendix}}, only $\backslash${\tt{section*}}.
%If you have multiple appendixes use $\backslash${\tt{appendices}} then use $\backslash${\tt{section}} to start each appendix.
%You must declare a $\backslash${\tt{section}} before using any $\backslash${\tt{subsection}} or using $\backslash${\tt{label}} ($\backslash${\tt{appendices}} by itself
%starts a section numbered zero.)}
%\section{References Section}
%You can use a bibliography generated by BibTeX as a .bbl file.
%BibTeX documentation can be easily obtained at:
%http://mirror.ctan.org/biblio/bibtex/contrib/doc/
% The IEEEtran BibTeX style support page is:
%\section{Simple References}
%You can manually copy in the resultant .bbl file and set second argument of $\backslash${\tt{begin}} to the number of references
% (used to reserve space for the reference number labels box).

\end{document}